\documentclass[12pt]{amsart}
\usepackage[T1]{fontenc}
\usepackage{makecell}
\usepackage[colorlinks=true, linkcolor = blue, citecolor={blue}]{hyperref}
\usepackage[margin=1in]{geometry}

\newcommand{\Z}{\ensuremath{\mathbb{Z}}}
\newcommand{\R}{\ensuremath{\mathbb{R}}}
\newcommand{\C}{\ensuremath{\mathbb{C}}}
\newcommand{\B}{\textbf}
\newcommand{\bra}[1]{|{#1}\rangle}
\newcommand{\ket}[1]{\langle{#1}|}
\newcommand{\tr}[1]{\operatorname{Tr}{\left(#1\right)}}
\newcommand{\alg}[1]{\operatorname{Lie}_{\mathbb{R}}\left(#1\right)}
\newcommand{\supp}[1]{\operatorname{supp}\left(#1\right)}
\renewcommand{\span}{\operatorname{span}_{\mathbb{R}}}
\newcommand{\ord}[1]{\operatorname{ord}\left(#1\right)}

\newcommand{\h}{\mathfrak{h}}
\renewcommand{\b}{\textbf{b}}
\renewcommand{\tilde}{\widetilde}
\renewcommand{\a}{\mathfrak{a}}

\newtheorem{theorem}{Theorem}[section]

\newtheorem{lemma}[theorem]{Lemma}
\theoremstyle{definition}

\theoremstyle{definition}

\title{Locality Induced Non-Universality for Abelian Symmetries}
\author{Sarvagya Jain}
\address{Department of Mathematics, Indian Institute of Science, Bangalore, India}
\email{sarvagyajain.math@gmail.com}
\begin{document}
\maketitle
\begin{abstract}
    According to a well-known result in quantum computing, any unitary transformation on a composite system can be generated using $2$-local unitaries. Interestingly, this universality need not hold in the presence of symmetries. In this paper, we study the analogues of the non-universality results for all Abelian symmetries. 
\end{abstract}

\section{Introduction}
Consider a system composed of $m$ qubits. An operator acting on such a system is called $k$-local if it acts non-trivially on Hilbert spaces of at most $k$ qubits. A fundamental problem in quantum computing is generating unitary transformations on a composite system using local operators. 
In this regard, a fundamental result states that any unitary transformation on a composite system can be generated by composing $2$-local unitaries \cite{DiVin}.

Conservation laws often restrict a physical system. As evident from Noether's theorem \cite{Noether}, these restrictions can be captured by imposing certain symmetry conditions on the class of unitary operators under consideration. Knowing this, it becomes a natural problem to generate symmetry-restricted unitary transformations on a composite system using local operators that obey the same symmetry restrictions.

Surprisingly, the universality from before fails to hold in this situation. Let $G$ be a group, and $U$ be a unitary representation of $G$ corresponding to its action on the $m$ qubit system. An operator $A$ acting on this system is called $G$-symmetric if $U(g)AU(g)^{\dagger} = A$ for all $g\in G$. Let $\mathcal{V}_{k}^{G}$ be the group of unitary matrices generated by $k$-local, $G$-symmetric unitary matrices, that is, $$\mathcal{V}_k^G := \langle A:A\text{ Unitary, $k$-local and $U(g)AU(g)^{\dagger} = A$ for all $g\in G$}\rangle.$$

Marvian considered symmetries of the form $U(g)=u(g)^{\otimes m}$ for all $g\in G$ and showed that $\dim \mathcal{V}_m^G-\dim \mathcal{V}_k^G\geq |\operatorname{irreps}(m)|-|\operatorname{irreps}(k)|$, where $\operatorname{irreps}(l)$ is the set of inequivalent irreducible representations of $G$ occurring in the representation $\{u(g)^{\otimes l}: g\in G\}$ \cite[Theorem 13]{Marvian}.
From this, it immediately follows that the universality fails for continuous symmetries like $U(1)$ and $SU(2)$. In this paper, we build upon the work of Marvian.

\section{Preliminaries}
Let $\a_k^G$ be the  set of $k$-local, $G$-symmetric skew-hermitian matrices, that is, $$\a_k^G:=\{A: A^{\dagger}+A = 0, \;A\; k\text{-local},\;[A, U(g)] = 0 \text{ for all $g\in G$}\}.$$ 

Let $\h_k^G$ be the real lie algebra generated by $k$-local, $G$-symmetric skew-hermitian matrices, that is,
$$\h_k^G := \alg{\a_k^G}.$$

The following theorem allows us to reduce the problem of characterising $\mathcal{V}_k^G$ to a more tangible problem of characterising $\h_k^G$.

\begin{theorem}
\label{lie alg lie grp corr}
Let $V$ be a unitary operator acting on the $m$ qubit system. Then, $V \in \mathcal{V}_{k}^{G}$ if and only if, there exists $C \in \h_k^G$ such that $V=e^{C}$. In other words, $\mathcal{V}_{k}^{G}=e^{\mathfrak{h}_{k}^G}$. Additionally, the dimension of $\mathcal{V}_{k}^{G}$ as a manifold is equal to the dimension of $\h_k^G$ as a vector space (over the field $\mathbb{R}$), that is,
$\dim\mathcal{V}_{k}^{G}=\dim\h_k^G$.
\end{theorem}
\begin{proof}
    Follows from \cite[Proposition 1, Corollary 3]{Marvian}.
\end{proof}

We will also need a characterisation of $\h_k^{U(1)}$.

For $A\in M_2(\C)$, let $A_r := I\otimes I\otimes \dots\otimes I\otimes \underbrace{A}_{r^{\text{th}}\text{ qubit}}\otimes I\otimes\dots\otimes I$ and $A^{\B{b}}:= \prod_{j:b_j = 1}A_j$, where $\b\in \{0, 1\}^m$. Let $C_l:= \sum_{\B{b}\in \{0, 1\}^m: w(\B{b}) = l}Z^{\B{b}}$ for $l = 0, \dots, m$, where $w(\B{b})$ denotes the Hamming weight of $\b\in \{0, 1\}^m$ and $Z$ is the Pauli matrix $\begin{pmatrix}1&0\\0&-1\end{pmatrix}$. 

For an $m$ qubit system, it is easily seen that $$\h_m^{U(1)} = \span\{i(\bra{\B{b}}\ket{\B{b}^{\prime}}+\bra{\B{b}^{\prime}}\ket{\B{b}}), \bra{\B{b}}\ket{\B{b}^{\prime}}-\bra{\B{b}^{\prime}}\ket{\B{b}}: \B{b}, \B{b}^{\prime}\in \{0, 1\}^m\text{ and } w(\B{b})= w(\B{b}^{\prime})\}.$$
\begin{theorem}
\label{U(1) char}
Let $m$ be the number of qubits. Then $A\in \h_k^{U(1)}$ if and only if $A\in\h_m^{U(1)}$ and $\tr{AC_l} = 0$ for $l = k+1, \dots, m$. In particular, $$\dim\mathcal{V}_m^{U(1)}-\dim\mathcal{V}_k^{U(1)} = \dim\h_m^{U(1)}-\dim\h_k^{U(1)} = m-k.$$
\end{theorem}
\begin{proof}
    Follows from \cite[Theorem 15]{Marvian}.
\end{proof}

\section{Our Results}
In this work, we characterise $\h_k^{G}$ for an arbitrary Abelian group $G$ whose unitary representation has the form $U(g)=u(g)^{\otimes m}$ for all $g\in G$, where $m$ is the number of qubits in the system. From Theorem \ref{lie alg lie grp corr}, a corresponding characterisation for $\mathcal{V}_k^{G}$ follows.

The set $\{u(g):g\in G\}$ is a commuting family of unitary matrices. Therefore it is simultaneously diagonalisable, that is, there exists a $2\times 2$ unitary matrix $P$ such that $Pu(g)P^{\dagger}=\lambda(g)$ for all $g\in G$, where $\lambda(g)$ is a $2\times 2$ diagonal matrix with diagonal entries $\lambda_{1,1}(g),\lambda_{2,2}(g)\in \mathbb{S}^1$. Let $n(g):=\ord{\frac{\lambda_{2,2}(g)}{\lambda_{1,1}(g)}}$ and $L := \operatorname{LCM}\left(n(g):g\in G\right)$, where we use the convention that if $n(g) = \infty$ for some $g\in G$ or $\Big|\left\{n(g):g\in G\right\}\Big|=\infty$, then $L=\infty$. The following theorem is the paper's main result and gives the characterisation of $\h_k^{G}$ based on the value of $L$. The characterisation is also summarised in Table \ref{summary of characterisation}.
\begin{theorem}
\label{main theorem}
Let $G$ be an Abelian group. Let $P, L$ be as above. 
\begin{enumerate}
    \item[(i)] If $L = \infty$, then $A\in\h_k^G$ if and only if $A\in \h_m^G$ and $\tr{AC_l} = 0$ for $l = k+1, \dots, m$. In particular, $\dim\h_m^G-\dim\h_k^G = m-k$.
    \item[(ii)] If $L$ is finite and $L<k$, then $A\in\h_k^G$ if and only if $A\in\left(P^{\dagger}\right)^{\otimes m}\h_m^{U(1)}P^{\otimes m}$ and $\tr{AC_l} = 0$ for $l = k+1, \dots, m$. Thus, $$\dim\h_m^G-\dim \h_k^G = \sum_{r=0}^{L-1}\left(\sum_{\substack{0\leq j\leq m\\ j\equiv r\mod{L}}}\binom{m}{j}\right)^2-\binom{2m}{m}-k+m.$$ 
    \item[(iii)] If $L$ is finite and $L\leq k$, then 
    \begin{itemize}
        \item for $L$ even, we have $A\in \h_k^G$ if and only if $A\in\h_m^G$ and $\tr{AC_m} =\tr{AZ^{\otimes m}}= 0$. Thus, $\dim\h_m^G-\dim\h_k^G = 1$.
        \item for $L$ odd, we have $\h_k^G = \h_m^G$.
    \end{itemize}
\end{enumerate}
\end{theorem}

\begin{table}
\caption{Characterisation of operators in $\h_k^G$ for an Abelian group $G$}
\label{summary of characterisation}
\centering
\begin{tabular}{c|c|c}

                \Xhline{3.5\arrayrulewidth}\rule{0pt}{3ex}       &                                                                                                                                                     $A\in \h_k^{G}$      & $\dim\h_m^{G}-\dim\h_k^{G}$                                                                                              \\[2.5ex] \Xhline{3.5\arrayrulewidth}\rule{0pt}{5ex}    

$L = \infty$           & \begin{tabular}[c]{@{}c@{}}$A\in \h_m^{G}$ and \\ $\tr{AC_l} =0$ for $l = k+1, \dots, m$\end{tabular}                                                         & $m-k$                                                                                                            \\[5ex] \hline\rule{0pt}{5ex}
$k<L< \infty$          & \begin{tabular}[c]{@{}c@{}}$A\in\left(P^{\dagger}\right)^{\otimes m}\h_m^{U(1)}P^{\otimes m}$ and \\ $\tr{AC_l} = 0$ for $l = k+1, \dots, m$\end{tabular} & $\sum_{r=0}^{L-1}\left(\sum_{\substack{0\leq j\leq m\\j\equiv r\mod{L}}}\binom{m}{j}\right)^2-\binom{2m}{m}-k+m$ \\[5ex] \hline\rule{0pt}{5ex}
$L\leq k$ and $L$ odd  & $A\in \h_m^{G}$                                                                                                                                               & 0                                                                                                                \\[5ex] \hline\rule{0pt}{5ex}
$L\leq k$ and $L$ even & \begin{tabular}[c]{@{}c@{}}$A\in\h_m^{G}$ and \\ $\tr{AC_m} =\tr{AZ^{\otimes m}}= 0$\end{tabular}                                                             & 1                                                                                                                \\[5ex] \hline
\end{tabular}
\end{table}

It is interesting to note that whenever $L$ is odd and satisfies $L\leq k$, we have universality, that is, $\h_k^G=\h_m^G$. Additionally, our result tells us exactly which operators can be implemented, strengthening some of the earlier results due to Marvian \cite[Theorem 13, Corollary 1]{Marvian} for the case when $G$ is Abelian. It will be interesting to see if similar explicit characterisations exist in the non-Abelian setting.

\section{Sketch of the Proof}
In this section, we sketch the proof of Theorem \ref{main theorem}. 

We first show that it suffices to prove the result for the case $G=\Z/n\Z$, and as a consequence of Theorem \ref{U(1) char}, we may assume $n\leq k$. 

When $n$ is even, we show by a pigeonhole principle argument that an operator in $\h_k^G$ satisfies a non-trivial diagonal constraint in addition to being in $\h_m^G$. This shows that $\dim\h_m^G-\dim\h_k^G\geq 1$. Then, using an inductive argument, we show that the subspace of $\h_k^G$ consisting of diagonal matrices has co-dimension at most $1$ in the space of all diagonal matrices. We also show that $\h_m^G$ is equal to the real Lie algebra generated by $\h_k^G$ 
and the space of all diagonal matrices. This allows us to deal with the off-diagonal constraints satisfied by the elements of $\h_k^G$. Finally, we do a patching argument similar to \cite{Marvian}. We show that $[\h_m^G, \h_m^G]\subseteq \h_k^G$. We will also explicitly describe $[\h_m^G, \h_m^G]$ from which the result will follow upon adding (as vector spaces) $[\h_m^G, \h_m^G]$ with the subspace of $\h_k^G$ consisting of diagonal matrices.

The case when $n$ is odd is quite similar, except for the fact that we don't have any non-trivial linear constraint and the subspace of $\h_k^G$ consisting of diagonal matrices has co-dimension $0$ in the space of all diagonal matrices of the same size. In this case, we won't need the patching argument.

\section{Initial Reductions}
Set $\Z/n\Z:=\Z$ whenever $n=\infty$. With this convention and the notation from before, it is easy to see that 
$$\a_k^G = \left(P^{\dagger}\right)^{\otimes m}\bigcap_{g\in G}\a_{k}^{\Z/n(g)\Z}P^{\otimes m}.$$
Here the action of the cyclic group $\Z/n(g)\Z$ is determined by the generator going to $\begin{pmatrix}1&0\\0&\omega(g)\end{pmatrix}^{\otimes m}$, where $\omega(g) := \frac{\lambda_{2,2}(g)}{\lambda_{1,1}(g)}\in \mathbb{S}^1$. Furthermore, observe that $\bigcap_{g\in G}\a_k^{\Z/n(g)\Z} = \a_k^{\Z/L\Z}$, where the action $\Z/L\Z$ is determined by the generator going to $\begin{pmatrix}1&0\\0&\omega\end{pmatrix}^{\otimes m}$ for some $\omega\in \mathbb{S}^1$ with $\operatorname{ord}(\omega)=L$. Thus, from now on, we will only consider the case where $G=\langle g\rangle\cong \Z/n\Z$ with unitary representation $U$ of $G$ such that $U(g) = \begin{pmatrix}1&0\\0&\omega\end{pmatrix}^{\otimes m}$ for some $\omega\in \mathbb{S}^1$ with $\operatorname{ord}(\omega)=n$.

We deal with the case when $k<n$. When $n=\infty$, we have the following result. 
\begin{theorem}
\label{n infinite}
An operator $A\in \h_k^G$ if and only if $A\in\h_m^G$ and $\tr{AC_l} = 0$ for $l = k+1, \dots, m$. Thus, $\dim\h_m^G-\dim \h_k^G = m-k$.
\end{theorem}
\begin{proof}
Since $n=\infty$, therefore $A$ commutes with $\begin{pmatrix}1&0\\0&\omega\end{pmatrix}^{\otimes m}$ if and only if $A$ commutes with $\left(e^{i\theta Z}\right)^{\otimes m}$ for all $\theta\in \R$. Thus, the result follows immediately from Theorem \ref{U(1) char}.
\end{proof}

Lastly, we have the following result if $k<n<\infty$.
\begin{theorem}
\label{n finite, n>k}
Let $k<n<\infty$. Then $A\in \h_k^G$ if and only if $A\in\h_m^{U(1)}$ and $\tr{AC_l} = 0$ for $l = k+1, \dots, m$. Thus, $$\dim\h_m^G-\dim \h_k^G = \sum_{r=0}^{n-1}\left(\sum_{\substack{0\leq j\leq m\\ j\equiv r\mod{n}}}\binom{m}{j}\right)^2-\binom{2m}{m}-k+m.$$
\end{theorem}
\begin{proof}
Since $k<n$, therefore a $k$-local operator commutes with $\begin{pmatrix}1&0\\0&\omega\end{pmatrix}^{\otimes m}$ if and only if it commutes with $\left(e^{i\theta Z}\right)^{\otimes m}$ for all $\theta\in \R$. Thus, from Theorem \ref{U(1) char}, it follows that $A\in \h_k^G$ if and only if $A\in\h_m^{U(1)}$ and $\tr{AC_l} = 0$ for $l = k+1, \dots, m$. The set of all hermitian operators commuting with $\left(e^{i\theta Z}\right)^{\otimes m}$ for all $\theta\in \R$, that is, $\h_m^{U(1)}$ has dimension $\binom{m}{0}^2+\binom{m}{1}^2+\dots+\binom{m}{m}^2 = \binom{2m}{m}$. Using Theorem \ref{U(1) char}, we get that $\binom{2m}{m}-\dim\h_k^G = m-k$.
Finally, we note that $$\dim\h_m^G = \sum_{r=0}^{n-1}\left(\sum_{\substack{0\leq j\leq m\\ j\equiv r\mod{n}}}\binom{m}{j}\right)^2.$$
\end{proof}

\section{Cyclic Symmetries with $n\leq k$}
Now, we deal with the remaining case, when the generating operators are $k$-local for $n\leq k$. Interestingly, the characterisation, in this case, depends on the parity of $n$.

\begin{theorem}
\label{n finite k at least n}
Let $n<m$ and $n \leq k$.
\begin{enumerate}
    \item[(i)] For $n$ odd, we have $\h_m^G = \h_k^G$.
    \item[(ii)] For $n$ even, we have $A\in \h_k^G$ if and only if $A\in\h_m^G$ and $\tr{AC_m} =\tr{AZ^{\otimes m}}= 0$. In particular, $\dim\h_m^G-\dim\h_k^G = 1$.
\end{enumerate}
\end{theorem}
In the remainder of this section, we prove a series of lemmas from which Theorem \ref{n finite k at least n} will follow.
\subsection{Diagonal Constraints}
In the following lemma, we show that for $n$ even, $\dim\h_m^G-\dim\h_k^G\geq 1$ by showing $\tr{AZ^{\otimes m}} = 0$ for all $A\in \h_k^G$.
\begin{lemma}
\label{necessity}
Let $n<m$, $n$ even and $n\leq k$. Then $A\in\h_k^G$ implies $$\tr{AC_m} =\tr{AZ^{\otimes m}}= 0.$$
\end{lemma}
\begin{proof}
Since $A\in\h_k^G$, therefore the operator $A$ commutes with $\left(\begin{pmatrix}1&0\\0&\omega\end{pmatrix}^{\frac{n}{2}}\right)^{\otimes m} = Z^{\otimes m}$. Let $A\in \a_k^G$, then by the pigeon hole principle, there exists $j\in \{1, \dots, m\}$ such that $A$ acts trivially on qubit $j$. Thus, we can write $A$ as linear combination of terms of the form $A^\prime \otimes I\otimes A^{\prime\prime}$, where $A^\prime$ acts on first $j-1$ qubits and $A^{\prime\prime}$ acts on last $n-j$ qubits. As $$\tr{\left(A^\prime \otimes I\otimes A^{\prime\prime}\right)Z^{\otimes m}}=\tr{A^\prime Z^{\otimes (j-1)}}\tr{Z}\tr{A^{\prime\prime}Z^{\otimes (n-j)}} = 0,$$ therefore $\tr{AZ^{\otimes m}} = 0$.  To finish the proof, we show that this property is also closed under $[\cdot, \cdot]$. Let $D, E\in \h_k^G$ have the desired property. As $D, E$ commute with $Z^{\otimes m}$, therefore $\tr{[D, E]Z^{\otimes m}} = \tr{[D, EZ^{\otimes m}]} = 0$.
\end{proof}

Our next result shows that for $n$ odd, $\h_k^G$ contains all the diagonal operators in $\h_m^G$.

Observe that $\left\{iZ^{\b}\right\}_{\b\in\{0, 1\}^m}$ forms a basis for the diagonal operators in $\h_m^G$. Additionally, for $\b, \b^\prime\in\{0, 1\}^m$, $$\frac{1}{2^m}\tr{Z^{\b}Z^{\b^\prime}} = \begin{cases}0&\text{if }\b\neq \b^\prime\\
1&\text{if }\b = \b^{\prime}\end{cases}.$$

For $\b\in\{0, 1\}^m$, let $\supp{\b}:=\{j:b_j=1\}$. For $b\in\{0, 1\}$, let $\neg b$ be its negation.
\begin{lemma}
\label{diag characterisation n odd}
Let $n<m$, $n$ odd and $n\leq k$. Then $iZ^{\b}\in \h_k^G$ for all $\b\in \{0, 1\}^m$.
\end{lemma}
\begin{proof}
For $\b\in\{0, 1\}^m$ with $w(\b) = n$, let $\tilde{\b} := b_1\dots b_{j-1}(\neg b_j)b_{j+1}\dots b_m\in \{0,1\}^m$, where $j$ is the largest index for which $b_j = 1$. Let
$$A_{\b} := i\begin{pmatrix}1&0\\0&0\end{pmatrix}^{\tilde{\b}},\;
\alpha_{\b} := \frac{i}{2}\left(\begin{pmatrix}0&1\\0&0\end{pmatrix}^{\b}+\begin{pmatrix}0&0\\1&0\end{pmatrix}^{\b}\right)\text{ and }
\beta_{\b} := \frac{1}{2}\left(\begin{pmatrix}0&1\\0&0\end{pmatrix}^{\b}-\begin{pmatrix}0&0\\1&0\end{pmatrix}^{\b}\right).$$

Observe that
\begin{itemize}
    \item $A_{\b} \in \span\{iZ^{\b}:\b\in\{0, 1\}^m, w(\b) <n\}$,
    \item $\alpha_{\b}, \beta_{\b}\in\h_n^G$,
\item $\left[A_{\b}, \alpha_{\b}\right] = -\beta_{\b}$, $\left[A_{\b}, \beta_{\b}\right] = \alpha_{\b}$, and
\item $\left[\alpha_{\b}, \beta_{\b}\right] = -\frac{i}{2^m}Z^{\b}+\gamma_{\b}$ for some $\gamma_{\b}\in \span\{iZ^{\b}:\b\in\{0, 1\}^m, w(\b) <n\}$.
\end{itemize}

We show that $iZ^{\b}\in \h_k^G$ for all $\b\in \{0, 1\}^m$ by induction on $w(\b)$. The result holds for $w(\b) = 0, \dots, n$ as $k\geq n$, establishing the base cases. Suppose that the result holds for all $\b\in\{0,1\}^m$ with $w(\b)<l$, where $l>n$. Let $\b\in\{0, 1\}^m$ be such that $\supp{\b} = \{j_1, \dots, j_l\}$. Let $\b_1, \b_2\in \{0, 1\}^m$ be such that $\supp{\b_1} = \{j_1, \dots, j_n\}$ and $\supp{\b_2} = \{j_{n+1}, \dots, j_l\}$. By the induction hypothesis, $A_{\b_1}Z^{\b_2}\in \h_k^G$. Thus, $\left[A_{\b_1}Z^{\b_2}, \beta_{\b_1}\right] = \alpha_{\b_1}Z^{\b_2}\in \h_k^G$. This implies that $$\left[\alpha_{\b_1}Z^{\b_2}, \beta_{\b_1}\right] = -\frac{i}{2^m}Z^{\b}+\gamma_{\b_1}Z^{\b_2}\in \h_k^G.$$ By the induction hypothesis, $\gamma_{\b_1}Z^{\b_2}\in \h_k^G$. Therefore, $iZ^{\b}\in \h_k^G$. This completes the induction and hence the proof.
\end{proof}

The next result characterises the diagonal operators in $\h_k^G$ for the case when $n$ is even.

\begin{lemma}
\label{diag characterisation n even}
Let $n<m$, $n$ even and $n \leq k$. Then $iZ^{\b}\in \h_k^G$ for all $\b\in \{0, 1\}^m$ with $w(\b)<m$.
\end{lemma}
\begin{proof}
For $\b\in\{0, 1\}^m$ with $w(\b) = n$, let
$$A_{\b} := \frac{i}{2}\left(\begin{pmatrix}1&0\\0&0\end{pmatrix}^{\b}-\begin{pmatrix}0&0\\0&1\end{pmatrix}^{\b}\right),\;
\alpha_{\b} := \frac{i}{2}\left(\begin{pmatrix}0&1\\0&0\end{pmatrix}^{\b}+\begin{pmatrix}0&0\\1&0\end{pmatrix}^{\b}\right)\text{ and }$$
$$\beta_{\b} := \frac{1}{2}\left(\begin{pmatrix}0&1\\0&0\end{pmatrix}^{\b}-\begin{pmatrix}0&0\\1&0\end{pmatrix}^{\b}\right).$$

For $\B{c}, \B{d}\in\{0,1\}^m$, we write $\B{c}\prec \B{d}$ if $\supp{\B{c}}\subseteq \supp{\B{d}}$.

Observe that
\begin{itemize}
    \item $A_{\b}=\frac{i}{2^m}\sum_{\substack{ \B{d}\prec \b\\ w(\B{d})\text{ odd}}}Z^{\B{d}}$,
    \item $\alpha_{\b}, \beta_{\b}\in\h_n^G$, and
\item $\left[A_{\b}, \alpha_{\b}\right] = -\beta_{\b}$, $\left[A_{\b}, \beta_{\b}\right] = \alpha_{\b}$, $\left[\alpha_{\b}, \beta_{\b}\right] = -A_{\b}$.
\end{itemize}

Again, we show that $iZ^{\b}\in \h_k^G$ for all $b\in \{0, 1\}^m$ with $w(\b)<m$ by induction on $w(\b)$. The result holds for $w(\b) = 0, \dots, n$ as $k\geq n$, establishing the base cases. Suppose that the result holds for all $\b\in\{0,1\}^m$ with $w(\b)<l$, where $m>l>n$.  Let $J:=\{j_1, \dots, j_{l+1}\}\subseteq \{1, \dots, m\}$. Let $\b_1, \b_2\in \{0, 1\}^m$ such that $\supp{\b_1} = \{j_1, \dots, j_n\}$ and $\supp{\b_2} = \{j_{n+1}, \dots, j_l\}$. By the induction hypothesis, $A_{\b_1}Z^{\b_2}\in \h_k^G$. Thus, $\left[A_{\b_1}Z^{\b_2}, \beta_{\b_1}\right] = \alpha_{\b_1}Z^{\b_2}\in \h_k^G$. This implies, $\left[\alpha_{\b_1}Z^{\b_2}, A_{\b_1}Z_{j_{l+1}}\right] = \beta_{\b_1}Z^{\b_2}Z_{j_{l+1}}\in\h_k^G$.  Therefore, $\left[\beta_{\b_1}Z^{\b_2}Z_{j_{l+1}}, \alpha_{\b_1}\right] = A_{\b_1}Z^{\b_2}Z_{j_{l+1}}\in\h_k^G$. By the induction hypothesis, 

$$\frac{i}{2^m}\left(\sum_{\substack{\B{d}\prec \b_1\\ w(\B{d})=n-1}}Z^{\B{d}}\right)Z^{\b_2}Z_{j_{l+1}}\in\h_k^G.$$

For a set $S\subseteq\{1, \dots, m\}$, let $\B{1}_S\in \{0, 1\}^m$ be such that $\supp{\B{1}_S} = S$. By permuting $j_1\dots, j_{l+1}$, we conclude that for any $A\subseteq J$ with $|A| = n$,

$$\frac{i}{2^m}\left(\sum_{\substack{\B{d}\prec \B{1}_A\\ w(\B{d})=n-1}}Z^{\B{d}}\right)Z^{\B{c}_A}\in\h_k^G,$$

where $\B{c}_A\in\{0,1\}^m$ with $\supp{\B{c}_A} = J\setminus A$. 

Let $\b\in\{0, 1\}^m$ with $\supp{\b}= J\setminus\{j_p\}$.  Observe that

$$\frac{i}{2^m}\left(\sum_{\substack{A\subseteq J\\ |A| = n\\ j_p\in A}}\left(\sum_{\substack{ \B{d}\prec \B{1}_A\\ w(\B{d})=n-1}}Z^{\B{d}}\right)Z^{\B{c}_A}\right)\in\h_k^G.$$

Rearranging, we get
$$\binom{l}{n-1}iZ^{\b}+\binom{l-1}{n-2}\left(\sum_{\substack{\B{d}\\w(\B{d})=l\\ j_p\in\supp{\B{d}} }}iZ^{\B{d}}\right)\in\h_k^G.$$

Additionally,
$$\frac{i}{2^m}\left(\sum_{\substack{A\subseteq J\setminus\{j_p\}\\|A| = n }}\left(\sum_{\substack{ \B{d}\prec \B{1}_A\\ w(\B{d})=n-1}}Z^{\B{d}}\right)Z^{\B{c}_A}\right)\in\h_k^G.$$

Rearranging, we get
$$\left(\binom{l}{n-1}-\binom{l-1}{n-2}\right)\left(\sum_{\substack{\B{d}\\w(\B{d})=l\\ j_p\in\supp{\B{d}} }}iZ^{\B{d}}\right)\in\h_k^G.$$

Thus, $iZ^{\b}\in\h_k^G$ for $\b\in\{0, 1\}^m$. Since the choice of $\{j_1, \dots, j_{l+1}\}$ and $\{j_p\}$ was arbitrary to begin with, therefore $iZ^{\b}\in\h_k^G$ for all $\b\in\{0, 1\}^m$ with $w(\b) = l$. This completes the induction and hence the proof.
\end{proof}
\subsection{Off-Diagonal Constraints}
The following lemma allows us to capture the off-diagonal constraints.
\begin{lemma}
\label{h_m characterisation}
Let $n<m$ and $n\leq k$. Then $\h_m^G = \alg{\{i\bra{\B{b}}\ket{\B{b}}:\B{b}\in \{0, 1\}^m\}\cup \h_k^G}$.
\end{lemma}
\begin{proof}
First note that
$$\h_m^G = \span\{i(\bra{\B{b}}\ket{\B{b}^{\prime}}+\bra{\B{b}^{\prime}}\ket{\B{b}}), \bra{\B{b}}\ket{\B{b}^{\prime}}-\bra{\B{b}^{\prime}}\ket{\B{b}}: \B{b}, \B{b}^{\prime}\in \{0, 1\}^m\text{ and } w(\B{b})\equiv w(\B{b}^{\prime})\mod n\}.$$

For $\b, \b^{\prime}\in\{0, 1\}^m$ such that $\b\neq \b^{\prime}$, $[i\bra{\B{b}}\ket{\B{b}}, \bra{\B{b}}\ket{\B{b}^{\prime}}-\bra{\B{b}^{\prime}}\ket{\B{b}}] = i(\bra{\B{b}}\ket{\B{b}^{\prime}}+\bra{\B{b}^{\prime}}\ket{\B{b}})$. Therefore \[\h_m^G = \alg{\{i\bra{\B{b}}\ket{\B{b}}, \bra{\B{b}}\ket{\B{b}^{\prime}}-\bra{\B{b}^{\prime}}\ket{\B{b}}: \B{b}, \B{b}^{\prime}\in \{0, 1\}^m\text{ and } w(\B{b})\equiv w(\B{b}^{\prime})\mod n\}}.\]

Let $\mathfrak{g} := \alg{\{i\bra{\B{b}}\ket{\B{b}}:\B{b}\in \{0, 1\}^m\}\cup \h_k^G}$. For $\B{b}, \B{b}^{\prime}\in\{0, 1\}^m$, let $F(\B{b}, \B{b}^{\prime}):=\bra{\B{b}}\ket{\B{b}^{\prime}}-\bra{\B{b}^{\prime}}\ket{\B{b}}$. Hence it suffices to show that $F(\b, \b^\prime)\in \mathfrak{g}$ for all $\b, \b^\prime\in \{0, 1\}^m$ such that $w(\b)\equiv w(\b^\prime)\mod{n}$.  

Observe that for $\b, \b^{\prime}, \b^{\prime\prime}\in\{0, 1\}^m$ such that $\b, \b^{\prime\prime}\neq \b^{\prime}$, $F(\B{b}, \B{b}^{\prime\prime})=\left[F(\B{b}, \B{b}^{\prime}), F(\B{b}^{\prime}, \B{b}^{\prime\prime})\right]$ (\textit{transitivity property}).

If $b_r\neq b_s$, then $[i\frac{X_{r}X_s+Y_rY_s}{2}, i\bra{\B{b}}\ket{\B{b}}] = F(\B{b}^{\prime}, \B{b})$, where $\B{b}^{\prime}$ is obtained $\B{b}$ by swapping the bits $r$ and $s$. Using this along with the \textit{transitivity property}  and the fact that transpositions generate the entire permutation group \cite{DF}, we conclude that $F(\B{b}, \B{b}^{\prime})\in \mathfrak{g}$ for $\B{b}$ and $\B{b}^{\prime}$ differing by a permutation of bits, i.e., for all $\b, \b^\prime\in \{0, 1\}^m$ such that $w(\b) = w(\b^\prime)$. 

For $\B{d}\in\{0, 1\}^m$ with $w(\B{d}) = n$, let $\alpha_{\B{d}}:=\frac{i}{2}\left(\begin{pmatrix}0&1\\0&0\end{pmatrix}^{\B{d}}+\begin{pmatrix}0&0\\1&0\end{pmatrix}^{\B{d}}\right)\in \h_{k}^G$. 

Let $\B{b}\in\{0, 1\}^m$ with $b_{r_1} = \dots=b_{r_n} = 0$. Let $\B{d}\in\{0, 1\}^m$ be such that $\supp{\B{d}} =\{r_1, r_2, \dots, r_n\}$. Then, $2[\alpha_{\B{d}}, i\bra{\B{b}}\ket{\B{b}}] = F(\B{b}^{\prime},\B{b})$, where $\B{b}^{\prime}\in\{0, 1\}^m$ is such that $\supp{\b^\prime} = \supp{\b}\cup\{r_1, \dots, r_n\}$.

Without loss of generality, let $\B{b}, \B{b}^{\prime}\in\{0, 1\}^m$ with $w(\B{b})\equiv w(\B{b}^{\prime})\mod{n}$ and $w(\B{b})< w(\B{b}^{\prime})$. We show that $F(\b, \b^\prime)\in\mathfrak{g}$. First, keep on increasing the hamming weight by $n$ by replacing $n$ of $0$ bits with $1$ repeatedly by using the \textit{transitivity property} and the observation in the last paragraph to get $F(\B{b},\B{c})\in \mathfrak{g}$ for some $\B{c}$ with $w(\B{c}) = w(\B{b}^{\prime})$. As $F(\B{c}, \b^\prime)\in \mathfrak{g}$, therefore by the \textit{transitivity property}, $F(\B{b},\B{b}^\prime)\in \mathfrak{g}$. 

Thus, $F(\B{b},\B{b}^\prime)\in \mathfrak{g}$ for all $\b, \b^\prime\in \{0, 1\}^m$ such that $w(\b)\equiv w(\b^\prime)\mod{n}$. This completes the proof.
\end{proof}

Using Lemma \ref{diag characterisation n odd} and \ref{h_m characterisation}, we get the part of Theorem \ref{n finite k at least n} concerning odd values of $n$.
\subsection{Patching Argument}
To finish the proof, it remains to patch together the diagonal constraints with the off-diagonal ones to characterise all the elements of $\h_k^G$ for $n$ even. 

Define $\Pi_l:=\sum_{\B{b}\in \{0, 1\}^m: w(\B{b}) \equiv l\mod{n}}\bra{\B{b}}\ket{\B{b}}$ for $l = 0, \dots, n-1$. Observe that the elements of $\h_m^G$ are block-diagonal with respect to $\{\Pi_l\}$.

\begin{lemma}
\label{off diag for n even}
Let $n<m$, $n$ even and $n \leq k$. Then $$\{A\in \h_m^G: \tr{A\Pi_l} = 0\text{ for }l=0,\dots, n-1\} = [\h_m^G, \h_m^G]\subseteq \h_k^G.$$ In particular, $\{X\in \h_m^G: X_{i,i}=0 \;\forall i\}\subseteq \h_k^G$. 
\end{lemma}
\begin{proof}
Let 
$\mathcal{D}:=\{A\in \h_m^G: \tr{A\Pi_l} = 0\text{ for }l=0,\dots, n-1\}$. Let $A, B\in \h_m^G$, then $$\tr{[A, B]\Pi_l} = \tr{[A, \Pi_l B]} = 0 \text{  for } l = 0,\dots,n-1.$$ Hence, $[\h_m^G, \h_m^G]\subseteq \mathcal{D}$. 

Now applying the fact that $[\mathfrak{su}(d), \mathfrak{su}(d)] = \mathfrak{su}(d)$ \cite{FH}
to blocks corresponding to each $\Pi_l$ separately, we conclude that for $A\in \mathcal{D}$, $\Pi_l A\Pi_l\in [\h_m^G, \h_m^G]$ for $l=0,\dots, n-1$. Using the fact that $A = \sum_{l=0}^{n-1}\Pi_lA\Pi_l$, we conclude that $A\in [\h_m^G, \h_m^G]$. Thus, $[\h_m^G, \h_m^G] = \mathcal{D}$.

From Lemma \ref{diag characterisation n even} and Lemma \ref{h_m characterisation}, we conclude that $\h_m^G = \alg{\{iZ^{\otimes m}\}\cup \h_k^G}$. As $iZ^{\otimes m}$ commutes with all elements of $\h_m^G$, therefore $[\h_m^G, \h_m^G] \subseteq \h_k^G$.
\end{proof}

Using Lemma \ref{diag characterisation n even} and \ref{off diag for n even}, we get the result for even values of $n$. This finishes the proof of Theorem \ref{n finite k at least n}.

\section{An Aside on the Effect of Different Representations}
\label{Different Reps}
In general, the unitary representation $U(\cdot)$ can act by different operators on each component, that is, $U(g) = u^{(1)}(g)\otimes u^{(2)}(g)\otimes\dots\otimes u^{(m)}(g)$, where the $u^{(j)}$'s need not be equal. 

In this section, we look at what happens for the action of $\Z/2\Z$ on the $m$ qubit system in this more general setting.  Let $G = \langle g\rangle\cong \Z/2\Z$ with unitary representation $U$ of $G$ such that $U(g) = u^{(1)}\otimes u^{(2)}\otimes\dots\otimes u^{(m)}$, where $u^{(1)}, \dots, u^{(m)}$ are $2\times 2$ unitary involutions. 

By spectral decomposition, there exists a $2\times 2$ unitary matrix $P^{(j)}$ such that $P^{(j)} u^{(j)} \left(P^{(j)}\right)^\dagger$ is equal to one of $I,-I,Z$ or $-Z$.

\begin{theorem} 
$\h_k^G = \h_m^G$ if and only if at most $k$ of $u^{(1)}, \dots, u^{(m)}$ are similar to $Z$ or $-Z$.
\end{theorem}

As $U(g) = PZ^{\B{b}}P^\dagger$ or $-PZ^{\B{b}}P^\dagger$, where $P:=\prod_{j=1}^m P_j^{(j)}$ and $\B{b}\in\{0, 1\}^m$ with $b_j = 1$ if and only if $P^{(j)} u^{(j)} \left(P^{(j)}\right)^\dagger = Z$ or $-Z$. Therefore, it suffices to prove the following lemma.

\begin{lemma}
$\h_m^G = \h_k^G$ if and only if  $U
(g) = Z^{\b}$ for some $\b\in\{0, 1\}^m$ with $w(\B{b})\leq k$.
\end{lemma}
\begin{proof}
Suppose $w(\B{b})>k$, then $\tr{Z^{\B{b}}A} = 0$ for all $A\in \h_k^G$ by an argument similar to the one given in Lemma \ref{necessity}. Thus, we have the implication in one direction.

For $\B{c}, \B{d}\in \{0, 1\}^{m}$, let $\B{c}\cdot{\B{d}},\neg \B{c}\in\{0, 1\}^m$ be defined such that $\supp{\B{c}\cdot \B{d}} = \supp{\B{c}}\cap\supp{\B{d}}$ and  $\supp{\neg \B{c}} = \{1, \dots, m\}\setminus \supp{\B{c}}$, respectively. 

To see the converse, let $w(\B{b})\leq k$. To generate all diagonal elements, it suffices to generate the elements of the set $\{iZ^{\B{d}}: w(\B{d})>k\}$. We will prove by induction that $iZ^{\B{d}}\in \h_k^G$ for all $d\in\{0, 1\}^m$ with $w(\B{d})\geq k$. By definition of $\h_k^G$, the base case follows. Let $\supp{\B{d}} = \{{l_1}, \dots, {l_t}\}$ for $t>k$. By pigeon hole principle, there exists $j\in \{1,\dots, t \}$ such that $Z^{\B{b}}$ acts trivially on qubit $l_j$. Without loss of generality, suppose $l_j = l_t$. By induction hypothesis, $iZ_{l_1}Z_{l_2}\dots Z_{l_{t-2}}Z_{l_{t}}\in \h_k^G$. Taking commutator of $iZ_{l_1}Z_{l_2}\dots Z_{l_{t-2}}Z_{l_{t}}$ with $iZ_{l_{t-1}}Y_{l_t}$, we conclude that $iZ_{l_1}Z_{l_2}\dots Z_{l_{t-1}}X_{l_t}\in \h_k^G$.  Taking commutator of $iZ_{l_1}Z_{l_2}\dots Z_{l_{t-1}}X_{l_t}$ and $iY_{l_t}$ gives us that $iZ^{\B{d}}\in \h_k^G$. This completes the induction.

Now to generate the off-diagonal elements, first note that
$$\h_m^G = \alg{\{i\bra{\B{c}}\ket{\B{c}}, \bra{\B{c}}\ket{\B{c}^{\prime}}-\bra{\B{c}^{\prime}}\ket{\B{c}}: \B{c}, \B{c}^{\prime}\in \{0, 1\}^m\text{ and } w(\B{c}\cdot{\b})\equiv w(\B{c}^{\prime}\cdot\B{b})\mod 2\}}.$$

We can apply the argument of Lemma \ref{h_m characterisation} restricted to $\supp{\b}$ to conclude that $$\{\bra{\B{c}}\ket{\B{c}^{\prime}}-\bra{\B{c}^{\prime}}\ket{\B{c}}: \B{c}, \B{c}^{\prime}\in \{0, 1\}^m, \B{c}\cdot{(\neg\B{b})} = \B{c}^{\prime}\cdot{(\neg\B{b})}\text{ and } w(\B{c}\cdot{\b})\equiv w(\B{c}^{\prime}\cdot\B{b})\mod 2\}\subseteq \h_k^G.$$ If $ b_j\neq 1$, then $iX_{j}\in \h_k^G$. Also note that $[iX_j, i\bra{\B{c}}\ket{\B{c}}] = \bra{\B{c}^{\prime}}\ket{\B{c}}-\bra{\B{c}}\ket{\B{c}^{\prime}}$, where $\B{c}^{\prime}$ is obtained from $\B{c}$ by flipping $c_j$, that is, $\B{c}^\prime$ has $\neg c_j$ instead of $c_j$. Using the fact that for $\B{c}, \B{c}^{\prime}, \B{c}^{\prime\prime}\in\{0, 1\}^m$ such that $\B{c}, \B{c}^{\prime\prime}\neq \B{c}^{\prime}$, $F(\B{c}, \B{c}^{\prime\prime})=\left[F(\B{c}, \B{c}^{\prime}), F(\B{c}^{\prime}, \B{c}^{\prime\prime})\right]$ repeatedly, where $F(\B{c}, \B{c}^{\prime}):=\bra{\B{c}}\ket{\B{c}^{\prime}}-\bra{\B{c}^{\prime}}\ket{\B{c}}$, with the aforementioned fact, we conclude that $\h_k^G = \h_m^G$. This completes the proof.
\end{proof}

The above result suggests that the universality shows behaviour that can be compared to a discrete analogue of phase transition. It will be interesting to see the generalisations of this behaviour to more complicated groups.

\section{Acknowledgements}
\label{Acknowledgements}
The author would like to thank Prof. Hadi Salmasian, University of Ottawa, for his guidance, insightful discussions throughout the summer of $2022$ and suggesting changes to the initial draft that helped improve the exposition. The author was supported by Mitacs GRI program during the research work.

\end{document}